\documentclass[review,12pt]{elsarticle}

\usepackage[latin1]{inputenc}
\usepackage{tikz}
\usetikzlibrary{arrows,positioning,automata,shadows,fit,shapes}
\usepackage{multirow}
\usepackage{epstopdf}
\usepackage{graphicx}
\usepackage{longtable}
\usepackage{amsthm}
\usepackage{amsmath}
\usepackage{amsfonts}
\DeclareSymbolFont{rsfscript}{OMS}{rsfs}{m}{n}
\DeclareSymbolFontAlphabet{\mathrsfs}{rsfscript}
\usepackage{amssymb}
\usepackage{bbm}
\usepackage{makeidx}
\usepackage{graphicx}
\usepackage[normalem]{ulem}
\usepackage{courier}

\newtheorem{teorema}{Theorem}

\newtheorem{lema}{Lemma}
\newtheorem{definicao}{Definition}

\begin{document}
	
	\begin{frontmatter}
		
		\title{The Reachability of Computer Programs}
		
		\author[inst1]{Reginaldo I Silva Filho}
		\ead{reginaldo.uspoli@gmail.com}
		
		\author[inst2]{Ricardo L Azevedo da Rocha\corref{au1}}
		\ead{rlarocha@usp.br}
		
		\author[inst1]{Camila Leite Silva}
		\ead{camila.leite002@gmail.com}
		
		\author[inst3]{Ricardo H Gracini Guiraldelli}
		\ead{ricardo.guiraldelli@univr.it}
		
		\address[inst1]{Department of Information Systems, Campus de Ponta Por\~a\\Universidade Federal de Mato Grosso do Sul - UFMS\\Rua Itibir\'e Vieira, s/n, BR 463, Km 4.5, 79907-414, Ponta Por\~a, MS, Brazil}
		
		\address[inst2]{Department of Computer Engineering, Escola Polit\'ecnica, Universidade de S\~ao Paulo\\Av. Luciano Gualberto, travessa 3, 380, 05508-900, Sao Paulo, SP, Brazil}
		
		\address[inst3]{Department of Computer Science, Universit\`{a} degli Studi di Verona\\Strada Le Grazie, 15, 37134, Verona, Italy}
		
		\cortext[au1]{Corresponding author}
		
		\begin{abstract}
			Would it be possible to explain the emergence of new computational ideas using the computation itself? Would it be feasible to describe the discovery process of new algorithmic solutions using only mathematics? This study is the first effort to analyze the nature of such inquiry from the viewpoint of effort to find a new algorithmic solution to a given problem. We define program reachability as a probability function whose argument is a form of the energetic cost (algorithmic entropy) of the problem.
			
			
		\end{abstract}
		
		
	\end{frontmatter}

\section{Introduction}

``The Golden Egg was not as exciting as the goose that laid it.''
This phrase, uttered by Ray J. Solomonoff \cite{solomonoff2011algorithmic}, is the motto of this paper. We want to investigate the possibility of explaining the discovery of computational ideas using only physics, and computation itself. The starting point of this investigation lies in the question:

``How difficult it is for a programmer discover a new algorithmic solution for a problem?''

It is necessary to explain in what sense the words ``algorithmic solution'' and ``problem'' we are using here. To clarify the meaning of these terms, we will take the Fibonacci's series as an example. The first ten elements of it are completely well-described by the string:

\begin{equation*}
0,1,1,2,3,5,8,13,21,34
\end{equation*} 

The program that prints the string above is an \textbf{algorithmic solution} for the problem. The question ``how to generate the string that represents the first ten terms of the Fibonacci series?'' represents this latter problem. Simplifying the problem; we considered the string itself as the actual \textbf{problem}. Thus, find an algorithmic solution is to discover a program for (to print) a problem (string).

The aim of this paper is to analyze the phenomenon of new programs  emergence (discovery by a programmer) in the physical world. We follow the  train of thought contained in \cite{deutsch1985quantum, deutsch2000machines}, which argues that computers are physical objects, computations are physical processes, and they exist in the real-world. The fact that such appearance occurs in our material universe and not in an entirely symbolic space, disassociated from any physical meaning, is our primary \textbf{postulate}. We will show that the logical consequence of this physicalist argument is that the probability of a programmer to develop a new program obeys an underlying entropy principle, informational at first, but also thermodynamics; thus leading us to the concept of \textbf{algorithmic reachability}.

As initial motivation, we concentrate on minimum length program probability. The objective is to associate the (incomputable) minimum length program concept with our reachability investigation. We will show that the emergence of minimum programs (determined by the
Kolmogorov complexity) is a function of the energetic intrinsic cost to achieve it. 

This paper is structured as follows. Section \ref{sec:Back} describes the notations and theoretical preliminaries; Section \ref{sec:Mot} presents the motivation of this work; Section \ref{sec:reachConcept}  presents the programs reachability, and Section \ref{sec:ReachFormulation} describes its formulation. The Section \ref{sec:reachResults} discusses the results presented in the previous section and the last section presents the conclusion. 
\section{Background}\label{sec:Back}

In this section, we give several definitions and notations required for the adequate discussion of the present article. We assume that the reader is familiar with basics concepts of physics, calculus, and  probability. 

\subsection{General definitions}

\begin{definicao}[\textbf{Lambert $W$ function}]
	Let $\mathbb{R}$ be the sets of all \textbf{real numbers}. For $x \in \mathbb{R}$, the Lambert $W$ function 
	\cite{borwein1999emerging}  is defined as the inverse of the function $f(x) = xe^{x}$ and solves : 
	\begin{equation}
	W(x) e^{W(x)} = x. \label{eq:LambertW}
	\end{equation}	
	$W(x)$ has the following behavior:
	\begin{itemize}
		\item For $x \geqslant 0$, $W(x)$ is a real positive function.
		\item For $- 1/e < x < 0$, $W$ is a multivalued application with two negative real-valued branches: 
		$W_0(x)$ and $W_{- 1}(x)$ \cite{corless1996lambertw}.
		\item For $x = - 1/e$, $W(x) = - 1$.
		\item For $x = 0$, $W(x) = 0$.
		\item For $x < - 1/e$, $W(x)$ is not defined in $\mathbb{R}$.
	\end{itemize}

\end{definicao}

\begin{figure}[h]
	
	\center
	
	\includegraphics[width=10cm]{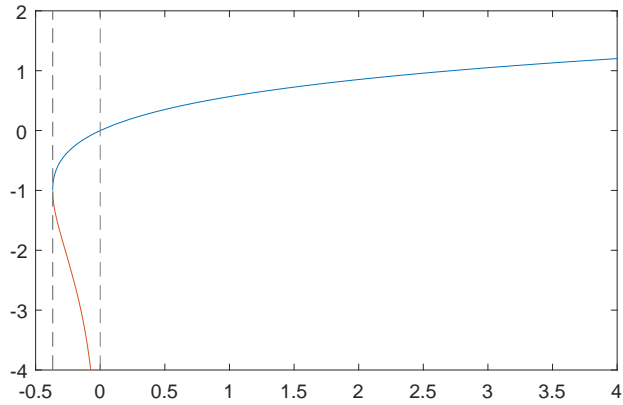}
	
	\caption{Graph of the Lambert W Function}
	
\end{figure}

\begin{definicao}[\textbf{Binary string and its length}]\label{def:Alfabeto}
	Let the binary alphabet  $B = \{0,1\}$ and $\mathbb{N}$ be the sets of all natural numbers. A \textbf{string}
	$\rho$ is any finite sequence of juxtaposed elements of the alphabet $B$ and its length $l(\rho)=n$
	is the number $n\in\mathbb{N}$ of the  elements composing $\rho$. 
	\end {definicao}
	\begin{definicao}[\textbf{Reflexive and Transitive Closure over the alphabet}]\label{def:RefTransClosure}
		Let $B^{k}=\{s:l(s)=k\}$ be the set of all strings with length $k$. The \textbf{reflexive and transitive closure over
			the alphabet} $B$ is defined as $B^{\ast}=\bigcup_{k=0}^{\infty}B^{k}$. In a similar way, 
		$B^{+}=\bigcup_{k=1}^{\infty}B^{k}$ defined as the \textbf{set of all nonempty strings over} $B$
	\end{definicao}
	\begin{definicao}[\textbf{Prefix String and Prefix free set}]\label{def:prefixString}
		Let the strings $r \in B^{\ast}$. In this context, a string $v$ is a \textbf{prefix of r} ($v \subseteq r $) if there exists
		$u \in B^{\ast}$ such that  $r = vu$. A set $F\subseteq B^{+}$ is called a \textbf{prefix-free set} when, $\forall (v, r \in F)$, $v \subseteq r$ is only true for $v = r$, . 
	\end{definicao}
	
	\subsection{Kolmogorov Complexity}
	
	\begin{definicao}[\textbf{Solution for $\rho$}] Given a universal Turing machine $\mathcal{U}$ and an input
		(program) in the form of a binary string $\varsigma$, if $\varsigma$ generates the string $\rho$ as output, so that:
		
		\begin{equation*}
		\mathcal{U}(\varsigma) = \rho
		\end{equation*}\\ then, we say that $\varsigma$ is a \textbf{solution} for the \textbf{problem $\rho$}.	
	\end{definicao}
	
	\begin{definicao}[\textbf{Chaitin machine and minimal program}]
		If a universal Turing machine $\mathcal{U}$ is defined for $\varsigma$ but not defined for any prefix $v$
		of $\varsigma$, then this machine is called a \textbf{Chaitin machine} \cite{solovay2000version} and denoted by $\mathfrak{U}$. Any
		$\varsigma$ accepted by $\mathfrak{U}$  is referred to as \textbf{a minimal program}. 
	\end{definicao} 
	
	\begin{definicao}[\textbf{$\rho$-solution set}]
		
		The set ${\Lambda}^{\rho } = \{\varsigma_{1}, \varsigma_{2},\dots, \varsigma_{n}\}$ of all minimal programs
		that are solutions of the problem $\rho$ is called a \textbf{$\rho$-solution set}. The set ${\Lambda}^{\rho }$
		is assumed to be finite.  
		
	\end{definicao}

	\begin{definicao}[\textbf{Kolmogorov complexity}]\cite[p.~110]{downey2010algorithmic}
		Given a problem $\rho$ with ${\Lambda}^{\rho}$, the Kolmogorov complexity of $\rho$ is defined as:  
		
		\begin{equation}\label{eq:kol}
		K(\rho) = min\{l(\varsigma_{i}):\mathfrak{U}(\varsigma_{i})=\rho \}	 
		\end{equation}	 
		Thus, the Kolmogorov complexity of $\rho$ is the $\rho$-solution set element with the lowest length. $\varsigma_{kol}$ now denotes this program. 
		
	\end{definicao} 
	
	\subsection {Boltzmann and Shannon entropy}
	
	\begin{definicao}[\textbf{Shannon's Entropy}]\label{def:shannonDef}\cite[p.~18]{gray1990entropy}
		For a finite sample space $\Omega$, let $X$ be a random variable taking values in $\Omega$ with probability
		distribution $p(x)$ for all $x \in X$. The \textbf{Shannon entropy} is defined as:
		\begin{equation}\label{eq:shannonEntr}
		H = - \sum_{x}p(x)\log_{2} p(x)
		\end{equation} 
		
	\end{definicao}	
	
	\begin{definicao}[\textbf{Boltzmann's Entropy}]
		For the special case of a discrete state space with L the counting measure defined as $L(\Gamma_x) = d(\Gamma_x)$, that is, the number of elements in the set $\Gamma_x$,
		\begin{equation}\label{eq:boltzEntr}
		SB(x) = (k \ln 2) \log d(\Gamma_x) 
		\end{equation}
		is the familiar form of the Boltzmann entropy \cite{li2013introduction}.
		
		The Gibbs entropy  is defined as: 	
		\begin{equation}\label{eq:gibbsEntr}
		SG = k \ln\Gamma 
		\end{equation}
		Where  $\Gamma$, which is the number of microstates corresponding to a given macrostate of a thermodynamic system \cite[p.~457]{kondepudi2014modern}, and  $k = 1.38065 \times 10^{-23}J/K$ is the Boltzmann constant.
	\end{definicao}
	
	\subsubsection{Relationship between Gibb's and Shannon's Entropy}.
	
	The relation between the Expression \ref{eq:shannonEntr} and \ref{eq:gibbsEntr} is expressed as follow \cite{cover2012elements,atlan1972organisation}:
	\begin{equation}\label{eq:relacShenTerm}
	SB = H k\ln2
	\end{equation}
	
	Following \cite[section 8.6]{li2013introduction}, the previous approaches (Gibbs, or Boltzmann's entropy and Shannon's entropy concepts) ``treated entropy as a probabilistic notion; in particular, each microstate of a system has entropy equal to $0$. However, it is desirable to find a concept of entropy that assigns a nonzero entropy to each microstate of the system, as a measure of its individual disorder.'' The last sentence means that Kolmogorov Complexity may be used to define the concept of algorithmic entropy.
	
	Considering that any measure of an individual microstate of a particular system has nonzero entropy and supposing that this system in equilibrium is described by the encoding $x$ of the approximated macroscopic parameters, one can estimate the entropy of the macrostate encapsulating the microstate. The algorithmic entropy of the macrostate of a system is given by $K(x) + H_x$, where $K(x)$ is the prefix complexity of $x$, and $H_x = \frac{SB(x)}{(k \ln 2)}$. Here SB(x) is the Boltzmann entropy of the system constrained by macroscopic parameters x, and k is Boltzmann's constant. The physical version of algorithmic entropy is defined as $SA(x) = (k \ln 2)(K(x) + H_x)$.
	
	Gibbs entropy is essentially the average algorithmic entropy. Let $H_{\mu}(x) = K(x) + O(1)$. Thus, the algorithmic entropy $H_{\mu}$ is a generalization of the prefix complexity $K$ \cite{li2013introduction}. The connection between Gibbs entropy and $H_{\mu}$ is given by the equation below \cite{li2013introduction}:
	\begin{equation}\label{eq:algEntropy}
	SA(w) = (k \ln 2) H_{\mu}(w), SA^{n}(w) = (k \ln 2) H_{\mu}^{n} (w)
	\end{equation}
	
	\subsection {Probability and physical work}
	
	\begin{definicao}[\textbf{Conditional Probability}]\cite[p.~22]{stark2002probability}
		Let $\Omega = (e_1, e_2, \dots, e_m)$ be a countable set of disjoint an exhaustive events, where each $e_i$
		is associated with a probability $P(e_i)$ such that $\sum_{i}P(e_i) = 1$ and $P(e_i)\neq 0$. Given $y$ as
		any event with $P(y)>0$, for all $i$: 
		\begin{equation}\label{eq:bayes}
		P(e_{i}|y) = \frac{P(y|e_{i})P(e_{i})}{P(y)}
		\end{equation}
		
	\end{definicao}	
	
	\textbf{The relationship physical work and Shannon entropy variation}. In a thermodynamic process \cite[p.~375]{sonntag1998fundamentals} with a transition taking a system from an initial state $1$ to some final state $2$, the work $\Delta W$ extracted in the course of the transition can be expressed regarding thermodynamic entropy as follow \cite{zurek1989algorithmic}:  
	
	\begin{equation}\label{eq:shannonWork1}
	\Delta W = T\Delta E
	\end{equation}
	
	Where $T$ is the temperature in Kelvin scale on Equation  \ref{eq:shannonWork1} which is valid for a process with a transition sufficiently slowly to be thermodynamically reversible and  with the internal energy of the  system considered  as constant. The variation $\Delta E$ is the difference between the entropy of the final and initial state\cite{lavenda2010new}:
	
	\begin{equation}
	\Delta E = S_2 - S_1
	\end{equation}
	
	By Expression \ref{eq:relacShenTerm}, we have:

	\begin{equation}\label{eq:shannonWork2}
	\Delta W = kT\ln2\Delta H
	\end{equation}
	
	
	\section{Motivation}\label{sec:Mot}
	
	In computer science, new programs are often designed. Some of them solve problems for which there was already a solution found. Some algorithms exist for centuries, some others are very recent, turning them part of the solution cluster for a determined class of computable problems. Others, on the other hand, provide more efficient solutions from the asymptotic analysis point of view, that is, on the amount of time or memory needed to get an output. 
	
	Regardless of the asymptotic complexity class, each computable problem has an algorithmic solution whose length in bits, prefix-free, is minimum, and that value is the \textbf{Kolmogorov complexity} equation~\ref{eq:kol} itself. The concept of Kolmogorov complexity relates to the principle of \textbf{Occam's Razor} \cite[p. 260]{li2013introduction}. We can interpret such principle as a method for selecting solutions, where the simplest case, among other explanations for the phenomenon under study, has to be chosen. Such principle presupposes the existence of a simplicity criterion in nature and Kolmogorov complexity appears as an objective measure to treat such simplicity. Unfortunately, for all string $\rho$, it is not possible to calculate $K(\rho)$, due to its incomputability.
	
	However, nothing prevents us from investigating the \textbf{occurrence probability} of $\varsigma_{kol}$, i.e. the probability that a programmer develops exactly the minimum program, the element of $\rho$-solution set ${\Lambda}^{\rho }$ with the smallest length in bits, whose length value is precisely the Kolmogorov Complexity $K(\rho)$.

	\subsection{Postulated Effects}
	
	In the introduction, we said that our main postulate is the fact that the discovery of new programs happens in the real-world. Although it is a postulate, we need to analyze most strictly the influence of this statement in our investigation of the minimum length program occurrence probability. 
	
	The expression of such probability involves the $\rho$-solution set ${\Lambda}^{\rho}$ and the derivation of Bayes' conditional rule (Expression \ref{eq:bayes}):
	
	\begin{equation}
	P(\varsigma_{kol}) = \frac{P(\varsigma_{kol}|\rho)\cdot P(\rho)}{P(\rho|\varsigma_{kol})}
	\label{eq:bayesappl}
	\end{equation}
	
	The postulated effects of the Expression \ref{eq:bayesappl} are given below:
	\paragraph{\textbf{Postulated Effect 1}}: the term $P(\rho|\varsigma_{kol})$ is the occurrence probability of  $\rho$  given the occurrence of $\varsigma_{kol}$. However it is necessary to put this term in a complete context: $P(\rho|\varsigma_{kol})$ is the probability of a Chaitin machine $\mathfrak{U}$, whose input is $\varsigma_{kol}$, generates the string $\rho$ as output. The program $\varsigma_{kol}$ is the minimum length  program that generates $\rho$. It generates $\rho$ and no other. It does not run forever, because $\varsigma_{kol}$ is an element of ${\Lambda}^{\rho}$ and not an arbitrary program. Thus, for $\mathfrak{U}(\varsigma_{kol})=\rho$,   $P(\rho|\varsigma_{kol}) = 1$. In this manner, Expression \ref {eq:bayesappl} takes the form :
	
	\begin{equation}
	P(\varsigma_{kol}) =  P(\varsigma_{kol}|\rho) \cdot P(\rho)
	\label{eq:bayes3}
	\end{equation}	
	
	\paragraph{\textbf{Postulated Effect 2}}: the term $P(\varsigma_{kol}|\rho)$ is the occurrence probability of $\varsigma_{kol}$ given the occurrence of $\rho$. At the same time, $\mathfrak{U}(\varsigma_{i})=\rho$ (for every $\varsigma_i \in {\Lambda}^{\rho}$), which means that $p(\varsigma_{kol}|\rho)$ represents the probability that the Chaitin machine $\mathfrak{U}$ have generated the
	string $\rho$ from the occurrence of the input $\varsigma_{kol}$. 	
	
	The construction (or reification) of a program comes from after  a problem attestation by a  witness. The fact that problem $\rho$
	can be described as a binary string is a consequence of it being noticed and somehow recorded by a group of witnesses that agree among themselves about its existence, and also how to represent it. Therefore, we can only talk about the situations
	where there is absolute certainty of the occurrence of $\rho$, that is, the situation where this occurrence is a prerequisite, where we have: 
	
	\begin{align} \label{post}
	&P(\varsigma_{kol}|\rho) = P(\varsigma_{kol})\\
	&P(\varsigma_{kol}|\neg \rho) = 0\\
	&P(\rho) = 1
	\end{align} 
	
	Thus, we have that the occurrence probability of $\varsigma_{kol}$ is always an \textit{a posteriori} probability, conditioned by the problem of generating $\rho$ as an output of $\mathfrak{U}$ by $\varsigma_{kol}$ input. 	 
	
	\section{The reachability of a program}\label{sec:reachConcept}
	
	However, what does to calculate the probability of a programmer to develop the program $\varsigma_{kol}$ mean, in a semantic sense?  Programs are not balls in a ballot box. They are not cards of a deck, dice faces or $\alpha$-particles shocking against a photographic plate. The Kolmogorov's complexity study about $\rho$ randomness nature is independent of the way $\varsigma_{kol}$ can be obtained \cite{grunwald2004shannon}. A raffle, a toss of a coin, or occurrence frequency calculus does not give the presence or absence of  $\varsigma_{kol}$, and it neither appears spontaneously in the input tape of the Chaitin machine $\mathfrak{U}$.
	
	The objective of a program is to solve a problem. There is a pragmatical necessity that the developed program generates a string $\rho$. Thus, the existence of $\varsigma_{kol}$ already implies (by its definition) the occurrence of a nontrivial phenomenon of compression that led to its description from $\rho$. To $\varsigma_{kol}$ be the input of $\mathfrak{U}$ there exist a sequence of events involved in the conception of the program and its constructive generation. This situation shows to us that the probability of a programmer developing a new algorithmic solution to a given problem correlates with the \textbf{necessity} notion, in an Aristotelian sense.
	
	Necessary is that which can not be otherwise \cite[p.~468]{Halper2008-HALOAM}. The necessity is what makes it impossible for something to be other than it is \cite[p.~120]{aristotle1998metaphysics}. Necessity is the reverse of the \emph{impossibility}. To illustrate the concept of impossibility, take the following illustration: a man can return to his childhood home and remodel it so that it looks, as much as possible, like what he has in mind, but he still cannot go back in the past. Hence, going backward to relive his past consists in an impossibility.
	
	The path followed by the programmer to find a new algorithmic solution lies under necessity constraints. In this sense, the Occam's razor is necessity principle, as well Shannon's entropy $H$ is a measure of the \emph{average} impossibility \cite[p.~141]{ben2008farewell}. Thus, there is a necessary way for a programmer to develop a new program and, at this moment, we need a clear connection between occurrence probability and necessity concept, to get the fact and showing how much it is necessary, what needs to happen and cannot be otherwise. This connection element is the $P(\varsigma_{kol})$ probability \ref{post}. We will not associate its meaning with a raffle phenomenon, but we will see it as a measure of the \textbf{program reachability}.
	
	
	\section{The program reachability calculus} \label{sec:ReachFormulation}
	
	In this way, we define the program reachability of $\varsigma_i \in \Lambda^{\rho}$ as the quantifiable relative necessity
	of an input $\varsigma_i$ for a Chaitin machine $\mathfrak{U}$ when $\varsigma_i$ generates $\rho$ as its output. Such relative necessity is represented by $P(\varsigma_{kol})$. However, a probability that cannot be expressed has no purpose. Thus, it is necessary to describe a way to obtain the program reachability expression regarding necessity or impossibility. The Shannon entropy, defined in
	\ref{def:shannonDef} will provide us this way. For this, let us take the random variable $X$ constituted by the enumeration elements of ${\Lambda}^{\rho}$. The image space of $X$ is expressed by:
	
	\begin{equation}
	I_X = \{1,2, \dots, i, \dots, m\}
	\end{equation}
	
	For such random variable, with respective probabilities $p(1),p(2),\dots,p(m)$, we have the Shannon entropy associated, given by the expression \ref{eq:shannonEntr}:
	
	\begin{equation}\label{eq:componente2Shann}
	H = - [p(1)\log_{2} p(1) + \dots +p(i)\log_{2} p(i)+\dots+p(m)\log_{2} p(m)]
	\end{equation} 
	
	Now, let us isolate the sum component concerning to $x = i$:
	
	\begin{equation}\label{eq:componente3Shann}
	H + [p(1)\log_{2} p(1) + \dots + p(m)\log_{2} p(m)] = - p(i)\log_{2} p(i)
	\end{equation} 
	
	At this point, it is important to note that the component
	
	\begin{equation}\label{eq:parcial result}
	[p(1)\log_{2} p(1) + \dots+p(i-1)\log_{2} p(i-1)+p(i+1)\log_{2} p(i+1)+\dots +p(m)\log_{2} p(m)]
	\end{equation}
	
	The right side of equation \ref{eq:componente3Shann} refers to the entropy of a random variable $X'$, which represents the enumeration of $A^\rho$, without its $i$th element. The  $X'$ space image description may be expressed as:
	
	\begin{equation}
	I_{X'} = I_X - \{i\}
	\end{equation}
	
	Such variable also has its own Shannon entropy, and the expression of this entropy is the component $[p(1)\log_{2} p(1) + \dots +p(m)\log_{2} p(m)]$,
	which will call $H'$. Thus, replacing the expression \ref{eq:componente3Shann}, we have:
	
	\begin{equation}\label{eq:componente4Shann}
	p(i)\log_{2} p(i) = -(H - H')
	\end{equation} 
	
	Another way to envision equation~\ref{eq:componente4Shann} is to consider it as describing the energy fluctuations by means of the difference on the entropy of the function that describes the program $\varsigma_{i}$, such that $p(i)\log_{2} p(i) = -\Delta(H_{\varsigma_{i}})$.
	
	Equation \ref{eq:componente4Shann} is non-linear. To determine $p(i)$ is necessary to deal with this non-linearity. Lambert $W$ function, described in Definition~\ref{eq:LambertW}, is an important mathematical tool which allows the analytic solutions for a broad range of mathematical problems \cite{brito2008euler}, including the solution of equation $x\log_{a} x = b$ (for $x, a, b \in \mathbb{R}$), which is \cite{corless1993lambert}:
	
	\begin{equation}
	x\log_{a} x = b \Longleftrightarrow x=e^{W(a\ln b)}
	\end{equation}
	
	Moreover, this is exactly the form of Equation \ref{eq:componente4Shann}. To simplify the analysis, we will denominate the expression $(H - H')$, the entropy variation, as $\mathrsfs{H}(\varsigma_{i})$. Notwithstanding, $p(i)$ represents the occurrence probability of the event $\varsigma_{i}$. Thus, with proper algebraic manipulation, \emph{for all} $\varsigma_i \in \Lambda^{\rho}$ (including the program $\varsigma_{kol}$)  we have:
	
	\begin{equation}\label{eq:componente6Shann}
	P(\varsigma_{i}) = \exp^{W(-\mathrsfs{H}(\varsigma_{i})\ln2)} 
	\end{equation}   
	
	As the entropy is positive, the domain of the function will be negative, which will yield a negative number for the $W$ function. Hence taking the $W_{-1}$ branch we'll rewrite it as:
	\begin{equation}\label{eq:probReach}
	P(\varsigma_{i}) = \exp^{-\|W_{-1}(-\mathrsfs{H}(\varsigma_{i})\ln2)\|}
	\end{equation} 
	
	Nevertheless we have an equation for $P(\varsigma_{i})$ that maps a semi-measure for the probability distribution over the finite event set, and it yields values ranging from $0 < P \leq 1$. Moreover, we can normalize it to become an actual measure for the probability stating: $P(\varsigma_{i}) = \frac{\exp^{-\|W_{-1}(-\mathrsfs{H}(\varsigma_{i})\ln2)\|}}{\sum_{i} P(\varsigma_{i})}$. In this case we actually have a probability distribution measure for the events, and $\sum_{i} P(\varsigma_{i}) \leq 1$.
	
	\section {Results discussion}\label{sec:reachResults}
	%
	
	\subsection{Bounds on reachability function domain}
	
	We need to analyze the behavior of $P(\varsigma_{i})$, impose no restrictions other than those typical of the analyzed function. Because the characteristics of Lambert $W$ function described in \ref{eq:LambertW}, the following domain restrictions for function $P(\varsigma_{i})$  are needed:
	
	\begin{equation}\label{eq:componente9Shann}
	-\dfrac{1}{e}\leq-\mathrsfs{H}(\varsigma_{i})\ln2 \leq 0
	\end{equation}\\ which implies:
	
	\begin{equation}\label{eq:componente11Shann}
	0\leq\mathrsfs{H}(\varsigma_{i}) \leq \dfrac{1}{e\ln2}
	\end{equation}

	Since $\frac{1}{e \ln 2} \approx 0.5307$, we may rewrite Inequality \ref{eq:componente11Shann} as

	\begin{equation*}
	0 \leq \mathrsfs{H}(\varsigma_{i}) \lessapprox 0.5307
	\end{equation*}

	\begin{figure}[h] 
		
		\center
		
		\includegraphics[width=10cm]{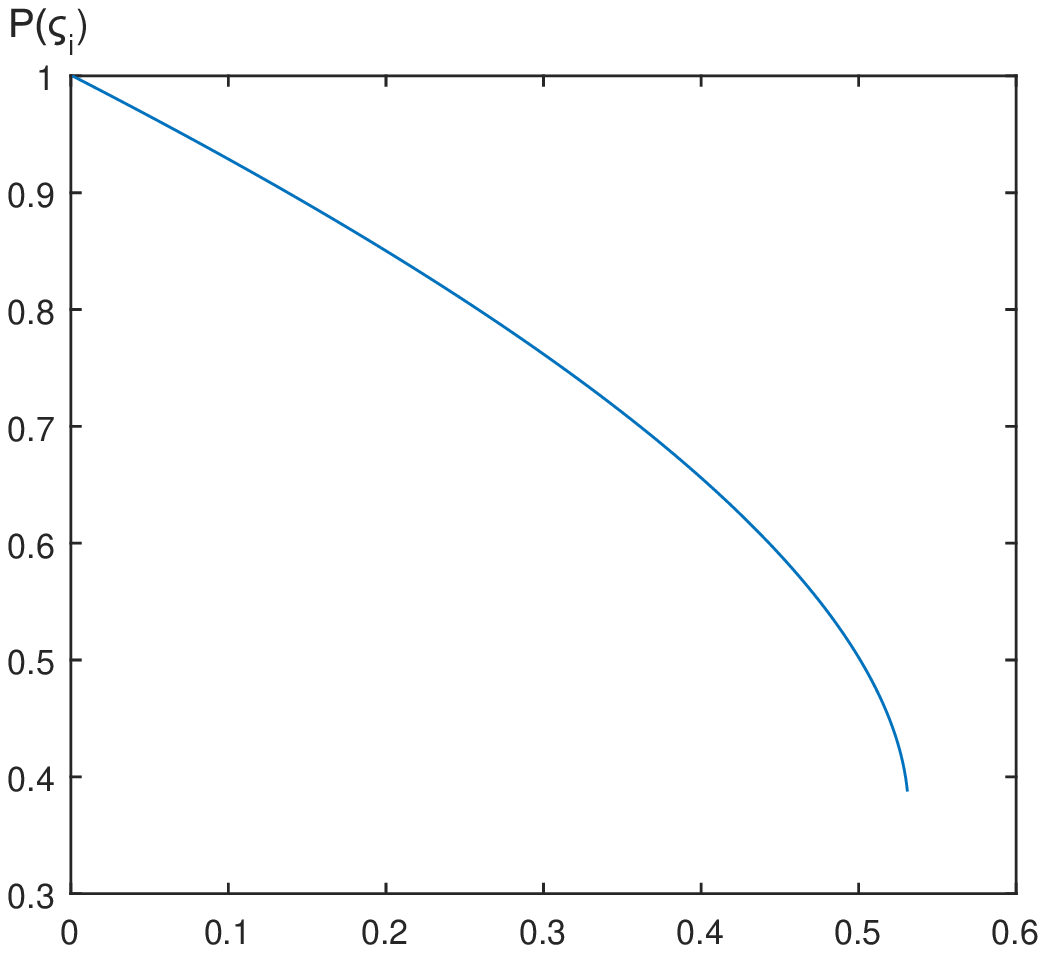}
		
		\caption{Graph of the function $P(\varsigma_{i})$. In $x$-axis, we have the values of $\mathrsfs{H}(\varsigma_{i})$}
		
	\end{figure}\label{fig:ProbGraph}


	Consider equations (\ref{eq:componente4Shann}) and (\ref{eq:probReach}) regarding a learning process for the algorithm reachability. We can apply equation (\ref{eq:probReach}) to reach a better program in a similar way the gradient descent algorithm is applied to machine learning algorithms. There is an algorithmic entropy associated with the process of learning \cite{li2013introduction,Baez2010,Betti2016}.
	
	One important feature of this approach on the Principle of Least Cognitive Action \cite{Betti2016} is that gradient descent algorithm can be derived from the principle.
	
	In this sense, we envision the exponential of the Lambert W function as a matching loss for the learning algorithm \cite{auer1996}, but it is not convex, besides being continuous, for the domain \cite[Proposition $5$]{borwein2016}.
	
	Since Lambert $W$ function is not convex, but continuous for the domain $x < 0$; our proposal for the probability semi-measure is $P(\varsigma_{i}) = \exp^{-\|W_{-1}(-\mathrsfs{H}(\varsigma_{i})\ln2)\|}$, which is bounded by the negative exponential of the Lambert's $W$ function. Proposition 7 from~\cite{borwein2016} deals with it:
	
	\begin{lema}\label{lem:convex}
		$\exp^{- W(x)}$ is convex. (This is proposition 7 from \cite{borwein2016})
	\end{lema}
	
	\begin{proof}
		See proof of proposition 7 on \cite{borwein2016}
	\end{proof}
	
	Hence we have:
	
	\begin{teorema}\label{th:loss}
		Let $\ell$ be a loss function as $\ell = \exp^{- W(x)}$, then $\ell$ is a matching loss.
	\end{teorema}
	
	\begin{proof}
		From Property $(P1)$ of~\cite{auer1996} we have this strong claim: $\ell(\hat{y},y): \mathbb{R} \times \mathbb{R} \rightarrow [0,\infty)$ is continuous and bounded. And from Theorem 5.1 of~\cite{auer1996} if a loss function follows the property then it is a matching loss.
		
		To prove the theorem we make:
		\begin{description}
			\item $y = f(z) = \exp^{ - W(z)}$
			\item $z = f^{-1}(y) = W^{-1}(\frac{1}{\exp^{-y^{-1}}}) = W^{-1}(\frac{1}{\ln(y)})$
			\item $\ell(\hat{y},y) = \ell(f(\hat{z}),f(z)) = \exp^{W(\hat{z})} - \exp^{W(z)} - f^{'}(z) (\hat{z} - z)$
		\end{description}
		which is convex in $\hat{z}$ by Lemma~\ref{lem:convex} (Proposition $7$ of \cite{borwein2016}). Hence, by Theorem 5.1 of~\cite{auer1996}, $\ell$ is a matching loss function.
	\end{proof}
	
	With these equations, our approach yields a theoretically sound method for the reachability calculus for a program. Nonetheless, Lambert W function is not entirely computable, for instance, $W(1) = \Omega$, Chaitin's incomputable constant.
	
	The simpler algorithm to reach a program builds all programs of the same size and checks them all; afterward, switches them to an immediately smaller size. On a subsequent step, the algorithm produces the smallest program only when there is no transformation left to reduce the size unless there is an informed hypotheses search-space. This process produces a meaningful result that is: the smallest program is not necessarily the most accessible one.
	
	Consider now that one method that defines an algorithm which is guided by the least energy cost given by Equation~\ref{eq:probReach}, applying a gradient descent. From this method we may argue its connection with the well-known Levin-Search algorithm \cite{levin1973,li2013introduction} since we also provide a search method (reachability), but for a program.
	
	\subsection{The Demiurge} 
	
	The connection between the concepts of probability and necessity does not eliminate the question of spontaneous generation of programs: after all, who is responsible for creating new programs, since their creation does not originate from throwing dices or tossing coins? Thus, beyond the concept of reachability, we need a constructive cause, a responsible agent, a being that we call \textbf{\textit{Demiurge}}. Consider this being as a version of Maxwell's demon. We can also understand the following digression from the arguments and measures exposed in \cite{berut2012experimental}.
	
	The philosopher Plato (in his dialog \textit{Timaeus}) \cite[p.~120]{ross1951plato} described the Demiurge (which in Greek means ``craftsman'') as a world-generation  entity, sometimes represented as endowed with only limited abilities. The Demiurge shapes the cosmos  via  imposing pre-existing form on matter, according to some ideal and perfect model \cite{o2015demiurge}. However, although it has such capacity, the Demiurge is not omnipotent, omnipresent and omniscient.
	
	The use of Demiurge in this paper has no intention of inserting any supernatural element in this discussion. The use of this term allows us to establish a cause for the generation of programs without the need of discussing the nature of this cause  (for instance, whether the Demiurge is a machine or not), only its functions. In principle, the Demiurge is a physical system.
	
	Thus, our Demiurge (represented by symbol $\mathfrak{D}$) has  the ability of ``to bring'' new algorithmic solutions for a problem $\rho$. However, the Demiurge  action   as the ``discoverer'' of a new solution for the problem $\rho$ involves an effort to carry out its function. There must be an energy consumption; which means that it needs to \textbf{realize work} to discover a new program $\varsigma_{i}$.
	
	Such fact has its effects on Expression \ref{eq:componente6Shann}, which states that the reachability of a program is a function of the entropy variation (represented by $\Delta H$ in Expression \ref{eq:shannonWork2} and, since Expression \ref{eq:componente6Shann}, represented by $\mathrsfs{H}(\varsigma_{i})$). However, we know (by Definitions \ref{eq:shannonWork1} and \ref{eq:shannonWork2}) that is possible to describe an entropy variation  in function of extracted work.
	
	\begin{equation}\label{eq:descoveryEntropy}
	\mathrsfs{H}(\varsigma_{i}) = \dfrac{\Delta W}{ kT\ln2}
	\end{equation}
	
	Similarly to what was made with the definition of $\mathrsfs{H}(\varsigma_{i})$, we will describe the variation of the work
	in terms of $\varsigma_{i}$ as follows:
	
	\begin{equation}\label{eq:workVar}
	\Delta W = \mathrsfs{E}(\varsigma_{i}) 
	\end{equation} 
	
	As the Demiurge $\mathfrak{D}$ is the responsible for this work and the variation of entropy is associated to the program reachability $\varsigma_{i}$ as well, the entropy variation of Expression \ref{eq:componente6Shann} we may substitute it by the right side of Expression \ref{eq:descoveryEntropy}.This results in:
	
	\begin{equation}\label{eq:workNew}
	P(\varsigma_{i}| \mathfrak{D}) = e^{W(-\mathrsfs{E}(\varsigma_{i})/{ kT})} 
	\end{equation}  
	
	The new version presented above makes the connection between the program reachability and the process of discovering new
	programs by the Demiurge $\mathfrak{D}$. Through this; we can observe that each solution to $\rho$ is associated with a determined quantity of energy expended to generate it. For the Expression \ref{eq:workNew}, the options for the energy values have their boundaries well determined (by Expression \ref{eq:componente11Shann}). Inside these limits, depending on the quantity of work  available (and the amount of work realized by the Demiurge), we may have a bigger or smaller probability of reaching a new solution.

	\section{Conclusion}\label{sec:conclusion}
	
	As stated in the introduction of this paper our goal was to set an abstraction for algorithmic discovery bounded by its reachability and its size. The idea behind the proposal is to view different programs simply as modified versions of the same basic program, this way any transformation between programs of the same size will not demand more energy to occur. However, to generate programs in a smaller size, it is necessary to apply work of the order $kT\ln2$ as stated on expression \ref{eq:descoveryEntropy}. This result agrees with the Landauer's principle \cite{berut2012experimental}.
	
	The simpler algorithm is then to build all programs of the same size and check them all, afterward switch them to an immediately smaller size. The next step of the algorithm produces the smallest program only when there is no transformation left to reduce the size unless there is an informed hypotheses search-space. This result means that generically speaking the smallest program not necessarily is the most accessible one.
	
	At this point the Demiurge cuts the knot by applying the work difference to the problem as an entropy result, thus producing an informed search on a non-informed hypotheses search-space. In fact using a gradient descent by Equation~\ref{eq:probReach}. We can envision this idea as applying a Kolmogorov complexity measure, using the Lambert $W$ function, at each transformation step of a program, seeking a path around the lower complexity programs.
	
	\subsection{Future Investigation}\label{sec:future-investigation}
	
	The presented results of this work span a series of connections to some areas that deserve further investigation. There are associations between temperature and entropy of a system in ~\cite[Section~2.1.11]{Wallace2005} through the thermodynamics's concept of microstates, and there is a study of these thermodynamics's concepts to algorithmic context~\cite{Baez2010}, but it is yet to be connected to this work.
	
	Still, in the line of states, the correlation between the current work and the minimum amount of states of a Turing machine must also be drawn. Such study is of interest not only because of the Kolmogorov complexity, one of the central arguments of this research, but also to define the theoretical lower-bound of Turing-machine states for any given program specification; this kind of information is useful, for instance, to study certain properties of specific algorithms such as the Busy Beaver~\cite{Yedidia2016}.
	
	Last but not the least, study ways to explore the Demiurge's informed search to develop programs with smaller sizes for a given output; although ambitious, it conceptually intersects with the works of Levin~\cite{levin1973,li2013introduction}, and on minimum message length~\cite{Wallace2005} and minimum description length~\cite{Gruenwald2007}.

	\section*{References}
	
	\bibliographystyle{plain} 
	\bibliography{refer}

\end{document}